\documentclass[12pt]{iopart}

\usepackage{iopams}
\expandafter\let\csname equation*\endcsname\relax
\expandafter\let\csname endequation*\endcsname\relax
\usepackage{amsmath}

\usepackage{mathrsfs,amsbsy,bm,amsthm,amsfonts,dsfont}
\usepackage[dvips]{graphicx}
\usepackage{times}
\usepackage{units}
\usepackage{braket}
\usepackage{amscd}
\usepackage{enumerate}
\usepackage{epsfig}
\usepackage{subfigure}
\usepackage{xcolor}

\newtheorem{theorem}{Theorem}
\newtheorem{lemma}{Lemma}

\begin{document}

\title{Quantum coherence via conditional entropy}

\author{Yunchao Liu$^1$, Qi Zhao$^1$ and Xiao Yuan$^{2,1}$}

\address{$^1$ Center for Quantum Information, Institute for Interdisciplinary Information Sciences, Tsinghua University, Beijing 100084, China}
\address{$^2$ Department of Materials, University of Oxford, Parks Road, Oxford OX1 3PH, United Kingdom}

\ead{yxbdwl@gmail.com}

\begin{abstract}
Quantum coherence characterizes the non-classical feature of a single party system with respect to a local basis. Based on a recently introduced resource framework, coherence can be regarded as a resource and be systematically manipulated and quantified.
Operationally, considering the projective measurement of the state in the computational basis, coherence quantifies the intrinsic randomness of the measurement outcome conditioned on all the other quantum systems. {However, such a relation is only proven {when randomness is characterized by} the Von-Neumann entropy.} In this work, we consider several recently proposed coherence measures and relate them to the general uncertainties of the projective measurement outcome conditioned on all the other systems. Our work thus provides a unified framework for {redefining} several coherence measures via general conditional entropies.
Based on the relation, we numerically calculate the coherence measures via semi-definite programming. Furthermore, we discuss the operational meaning of the unified definition. Our result highlights the close relation between single partite coherence and bipartite quantum correlation.
\end{abstract}



\maketitle

\section{Introduction}
Quantum coherence is one of the most fundamental non-classical features of quantum systems \cite{Aharonov67,Kitaev04,Bartlett07}. Considering a $d-$dimensional Hilbert space $\mathcal{H}$ in a computational basis $\mathbb{I}= \{\ket{i}\}_{i=0,\cdots, d-1}$, coherence characterizes the superposition property between different basis states. Given a pure state $\rho = \ket{\psi}\bra{\psi}$, where $\ket{\psi} = \sum_i a_i\ket{i}$ and $\sum_{i}|a_i|^2=1$, the amount of coherence $C(\rho)$ can be defined by the {Shannon entropy of the probability distribution $\{|a_i|^2\}$,  or alternatively $S(\Delta(\rho))$} with $\Delta(\rho) = \sum_{i}\ket{i}\bra{i}\rho\ket{i}\bra{i}$ being the dephased state of $\rho$ and $S(\rho) = -\Tr[\rho\log_2\rho]$ being the Von Neumann entropy. {From the perspective of quantum random number generation, quantum superposition in a local basis allows us to extract true randomness (i.e. secure random bits) from a quantum state by performing measurement in that basis. It is shown that the randomness of a pure state $\ket{\psi}$ with respect to the $\mathbb{I}$ basis is also given by $S(\Delta(\ket{\psi}\bra{\psi}))$ \cite{Ma2016QRNG, herrero2017}.} Thus, coherence not only characterizes the superposition in the computational basis but also quantifies the randomness or uncertainty if we measure the state in the same basis.
This definition cannot be directly generalized to mixed states as $C(\rho) = S(\Delta(\rho))$ by naively considering the whole randomness of measurement outcome. Such a definition will lead to a contradiction that the maximally mixed state even has the maximal amount of coherence $S\left(\frac{1}{d}\sum_i \ket{i}\bra{i}\right) = \log_2 d$. In general, if a state is only a mixture of the basis states, such as $\delta=\sum_{i}\delta_i\ket{i}\bra{i}$, it can be prepared in a classical way without involving any coherence or superposition, and hence produces no true randomness by measuring it in the $\mathbb{I}$ basis. Therefore, the coherence of $\delta$ should be zero in the $\mathbb{I}$ basis and we need a more sophisticated definition of coherence for general mixed states.

On the one hand, coherence frameworks \cite{aberg2006quantifying, baumgratz2014quantifying} are proposed  by considering coherence as a resource of superposition. Focusing on the $\mathbb{I}$ basis, we define {\it incoherent state} as $\delta=\sum_{i}\delta_i\ket{i}\bra{i}$ and consider it as zero resource state{, i.e. state with zero coherence}. The incoherent state set is denoted by $\mathcal{I}=\{\delta|\delta=\sum_i\delta_i\ket{i}\bra{i}\}$. Furthermore, the \emph{incoherent operation} is defined as a physically realizable operation that only transforms incoherent states to incoherent states. Based on the definitions of incoherent state and incoherent operation, a general resource theory of coherence is completed by defining coherence measures as a real-valued function $C(\rho)$ that satisfies several requirements. The advantage of this framework is that coherence can be studied in an abstract and systematical way. Coherence measures can be mathematically proposed and proven, such as the relative entropy of coherence \cite{baumgratz2014quantifying}, the coherence of formation \cite{aberg2006quantifying, yuan2015intrinsic}, and the robustness of coherence \cite{Napoli16}.
Furthermore, the {resource theory of coherence} supplies a basic framework for studying its relationship with general quantum correlations \cite{Yao15,Streltsov15,Streltsov16, ma16, Chitambar16,Hu17relative,PhysRevX.7.011024,yuan2017unified} and probing general coherence properties \cite{yu2016measure, Yu16Alternative,PhysRevX.6.041028,Shi17, hu2017maximum,shi2017coherence, Streltsov17structure,Zhou17,Yang2017}.
We refer to Ref.~\cite{streltsov2016quantum,Hu2017arXivreview} for reviews of recent developments of the resource theory of coherence.

On the other hand, coherence can be understood as the \emph{intrinsic randomness}  by measuring the state in the computational basis \cite{yuan2015intrinsic, yuan2016interplay, hayashi2017secure, Yuan17uncertainty, Luo17uncertainty, ma2017source}. Denote the system we study as $A$, then the coherence of $\rho_A$ quantifies the intrinsic randomness   of the measurement outcome in the computational basis. More rigorously, the intrinsic randomness  is defined as the unpredictable randomness {which is secure from the attack by} any other systems. Specifically, considering a system $E$ that holds the maximal information of $A$, i.e., a purification of state $\rho_A$, the intrinsic randomness  is given by the uncertainty of the measurement outcome conditioned on the information of system $E$.

In previous works, the uncertainty is characterized by the conditional Von-Neumann entropy. Here, we generalize the results to general conditional entropies. We consider {coherence measures including} the relative entropy of coherence \cite{baumgratz2014quantifying}, coherence of formation \cite{yuan2015intrinsic}, geometric coherence \cite{Streltsov15}, two {recently proposed} coherence measures \cite{Rastegin16, zhao2017one}, and a new coherence measure---{the min-entropy of coherence}. Then, we redefine them in a unified picture by {relating to} general conditional entropies. In the following, we will first review the coherence framework and coherence measures in Sec.~\ref{Sec:framework}. Then, we rephrase the coherence measures via the corresponding conditional entropy in Sec.~\ref{Sec:conditional}. In Sec.~\ref{Sec:SDP}, we propose a numerical method to calculate the coherence measures via semi-definite programming. We discuss the operational meaning of the new definitions in Sec.~\ref{Sec:operation} and summarize in Sec.~\ref{Sec:summary}.

\section{Coherence measures}\label{Sec:framework}
In this work, we mainly focus on the resource framework of coherence proposed in Ref.~\cite{baumgratz2014quantifying}.

\subsection{Coherence framework}
Considering a $d-$dimensional Hilbert space $\mathcal{H}$ with a computational basis $\{\ket{i}\}_{i=0,\cdots, d-1}$, the set of \emph{incoherent states} is defined by
\begin{equation}\label{incoherent}
  \mathcal{I}=\{\delta|\delta=\sum_{i=0}^{d-1}\delta_i\ket{i}\bra{i}\}.
\end{equation}
Meanwhile, the maximally coherent state is defined as
\begin{equation}\label{maximally}
  \ket{\psi_d}=\frac{1}{\sqrt{d}}\sum_{i=0}^{d-1}e^{i\phi_i}\ket{i},
\end{equation}
where $\phi_i$ is an arbitrary phase on basis $\ket{i}$.
Incoherent states can be understood as free states; while incoherent operations is similarly defined as free operations. In general, there are several different definitions of incoherent operations \cite{chitambar2016comparison}. In this work, we focus on the \emph{Incoherent Operation} (IO) proposed in Ref.~\cite{baumgratz2014quantifying}, which is defined as a CPTP map $\Lambda(\rho)=\sum_nK_n\rho K_n^\dag$ such that $K_n\delta K_n^\dag/\Tr[K_n\delta K_n^\dag]\in\mathcal{I}$ for all $n$ and $\delta\in\mathcal{I}$. Meanwhile, we also consider the \emph{Maximal Incoherent Operation} (MIO) \cite{aberg2006quantifying} that is defined as a CPTP map $\Lambda(\rho)$ such that $\Lambda(\delta)\in\mathcal{I}$ for all $\delta\in\mathcal{I}$. It follows from the definition that IO $\subseteq$ MIO.

In general, the amount of quantum coherence of a state $\rho$ is characterized by a nonnegative real-valued function $C(\rho)$ which satisfies the following properties:
\begin{enumerate}[(C1)]
	\item $C(\rho)\ge0,\forall \rho$ and $C(\delta) = 0$ iff $\delta\in\mathcal{I}$;
	\item Monotonicity: coherence cannot increase under MIO or IO map $\Lambda$, i.e., $C(\Lambda(\rho)) \le C(\rho)$;
	\item Strong monotonicity (with post-selection): {for any $\Lambda\in\text{IO}$ with Kraus operators $\{K_n\}$}, coherence cannot increase on average under post-selection, i.e., $\sum_n p_nC(\rho_n) \le C(\rho)$, where $\rho_n = {K_n\rho K_n^\dag}/{\mathrm{Tr}\left[ K_n\rho K_n^\dag\right]}$;
	\item Convexity: coherence cannot increase by mixing quantum states,  i.e., $C\left(\sum_n p_n\rho_n\right) \le \sum_n p_nC(\rho_n)$.
\end{enumerate}
Note that the strong monotonicity requirement (C3) is defined only for IO as the the Kraus operators of MIO are not explicitly specified by definition. In general, when $C(\rho)$ satisfies (C1) and either (C2) or (C3)  but not (C4), we still consider it as a coherence monotone as it may still play useful role in a practical task.

A general way of constructing coherence measures is by minimizing a distance function over all incoherent states, that is
\begin{equation}\label{distance}
  C(\rho)=\min_{\delta\in\mathcal{I}}D(\rho,\delta),
\end{equation}
where $D(\rho,\delta)$ is a distance function satisfying $D(\rho,\delta)=0$ if and only if $\rho=\delta$ \footnote{Note that we do not require the triangle inequality for the distance function by following the convention from Ref.~\cite{baumgratz2014quantifying}.}. For instance, considering the relative entropy as the distance function $D(\rho,\sigma)=S(\rho||\sigma) = \Tr[\rho\log_2\rho]-\Tr[\rho\log_2\sigma]$, we can define the relative entropy of coherence \cite{baumgratz2014quantifying},
\begin{equation}
	C_r(\rho)=\min_{\delta\in\mathcal{I}}S(\rho||\delta) = S(\Delta(\rho))-S(\rho).
\end{equation}
The relative entropy of coherence measures the asymptotic rate of coherence distillation under IO \cite{Winter16}. We refer to Appendix A for more information about asymptotic and one-shot coherence conversion.
Suppose the distance function is defined by $D(\rho,\sigma)=1-F(\rho, \sigma)$ where $F(\rho,\sigma)=\left(\Tr\left[\sqrt{\sqrt{\rho}\sigma\sqrt{\rho}}\right]\right)^2$ is the fidelity between $\rho$ and $\sigma$, we obtain the geometric coherence \cite{streltsov2015measuring},
\begin{equation}
	C_g(\rho)=\min_{\delta\in\mathcal{I}}\left(1-F(\rho, \delta)\right).
\end{equation}
Furthermore, considering the distance function as the max and min quantum Renyi divergence $D(\rho, \sigma) = D_{\max}(\rho||\sigma)=\log_2\min\{\lambda\in\mathbb{R}|\rho\leq\lambda\sigma\}$ and $D_{\min}(\rho||\sigma)=-\log_2 F(\rho,\sigma)$, respectively, we can define the max and min-entropy of coherence $C_{\max}(\rho)$ and $C_{\min}(\rho)$  as
\begin{equation}\label{cmaxcmin}
\begin{aligned}
	C_{\max}(\rho)&=\min_{\delta\in\mathcal{I}}D_{\max}(\rho||\delta),\\
	C_{\min}(\rho)&=\min_{\delta\in\mathcal{I}}D_{\min}(\rho||\delta).
\end{aligned}
\end{equation}
Here the max-entropy of coherence is used in Ref.~\cite{zhao2017one} to characterize the one-shot coherence dilution under MIO and the min-entropy of coherence is a new coherence measure that characterizes one-shot coherence distillation under IO \footnote{This work is under preparation.}.

Alternatively, another way of defining coherence measure is via the convex-roof construction.  For instance, the coherence of formation $C_f(\rho)$ \cite{aberg2006quantifying,yuan2015intrinsic}  can be defined by
\begin{equation}
\begin{aligned}
C_{f}(\rho)&\equiv\min_{\{p_j,\ket{\psi_j}\}}\sum_jp_jS(\Delta(\ket{\psi_j}\bra{\psi_j})).\\
\end{aligned}
\end{equation}
Here, the minimization is over all possible decomposition of the state $\rho$ as $\rho = \sum_j p_j\ket{\psi_j}\bra{\psi_j}$.
Coherence of formation measures the asymptotic rate of coherence dilution under IO \cite{Winter16}.
Suppose $\ket{\psi_j}=\sum_{i=0}^{d-1}a_{ij}\ket{i}$ and denote $T_j$ to be the number of nonzero elements in $\{a_{0j},\cdots,a_{d-1,j}\}$, we can also define another coherence monotone $C_0(\rho)$ as
\begin{equation}
	C_0(\rho)=\min_{\{p_j,\ket{\psi_j}\}}\max_{j}\log_2 T_j.
\end{equation}
The coherence monotone $C_0(\rho)$ measures the one-shot coherence dilution rate under IO \cite{zhao2017one}.

\subsection{Properties of coherence measures}
In this work, we focus on the six coherence measures $C_r$, $C_g$, $C_{\max}$, $C_{\min}$,  $C_f$, and $C_0$. The proofs that they satisfy the properties (C1) --- (C4) can be found in the corresponding references. Here, we only prove that the new coherence measure $C_{\min}$ satisfies the properties {(C1), (C2) and (C4)}.

\begin{theorem}
  The function $C_{\min}(\rho)$ is a coherence monotone satisfying (C1), (C2) and (C4).
\end{theorem}

The proof of the theorem is left in Appendix B.
In addition, we show that there exists an relationship among the six coherence measures or monotones.
\begin{theorem}\label{relationship}
  For any quantum state $\rho$, we have $C_g(\rho)\leq C_{\min}(\rho)\leq C_r(\rho)\leq C_{\max}(\rho) \le  C_0(\rho)$ and $C_r(\rho)\leq C_f(\rho)\leq C_0(\rho)$.
\end{theorem}
\begin{proof}
  First note that $C_g(\rho)=1-2^{-C_{\min}(\rho)}$ from the definitions of $C_g$ and $C_{\min}$. Since $1-2^{-x}\leq x$ for all $x\geq 0$, we have $C_g(\rho)\leq C_{\min}(\rho)$. Next, we use the monotonicity of the generalized $\alpha$-Renyi divergence
  \begin{equation}\label{alpha}
    \tilde{D}_{\alpha}(\rho||\sigma)=\frac{1}{\alpha-1}\log_2\left(\Tr\left[\left(\sigma^{\frac{1-\alpha}{2\alpha}}\rho\sigma^{\frac{1-\alpha}{2\alpha}}\right)^{\alpha}\right]\right)
  \end{equation}
which is non-decreasing with respect to $\alpha$ \cite{2013MartinRenyi}. Since $D_{\min}(\rho||\sigma), D(\rho||\sigma)$ and $D_{\max}(\rho||\sigma)$ correspond to $\alpha\to\frac{1}{2},1,\infty$, respectively, we know that $\forall \rho,\sigma$,
\begin{equation}\label{ineq}
  D_{\min}(\rho||\sigma)\leq D(\rho||\sigma)\leq D_{\max}(\rho||\sigma).
\end{equation}
Taking the minimization over $\delta\in\mathcal{I}$ on each side, we obtain $C_{\min}(\rho)\leq C_r(\rho)\leq C_{\max}(\rho)$. Furthermore, since $\textrm{IO}\subseteq\textrm{MIO}$ and $C_{\max}(\rho)$ and $C_0(\rho)$ correspond to the one-shot dilution rate under MIO and IO, respectively \cite{zhao2017one}, we know that $C_{\max}(\rho)\leq C_0(\rho)$.

For the second inequality, according to Ref.~\cite{Winter16}, $C_r(\rho)$ and $C_f(\rho)$ correspond to the asymptotic coherence distillation and dilution rate under IO, which indicates that $C_r(\rho)\le C_f(\rho)$. As $C_0(\rho)$ corresponds to the one-shot dilution rate under IO, we thus have $C_f(\rho)\leq  C_0(\rho)$.

We also note that, at the moment, the relationship between $C_f(\rho)$ and $C_{\max}(\rho)$ is unknown.
\end{proof}

\section{Coherence measures via conditional entropies}\label{Sec:conditional}
In this section, we rephrase all the six coherence measures in an unified framework via conditional entropies.

\subsection{Conditioned on quantum information}
Denote the system we study as $A$ which acts on a $d_A$ dimensional Hilbert space $\mathcal{H}_A$, and the quantum state as $\rho_A$ that is from the state set $\mathcal{D}(\mathcal{H}_A)$. We also consider another system $E$ and a larger state $\ket{\psi}_{AE}$ that purifies $\rho_A$, that is, $\rho_A=\Tr_E[\ket{\psi}_{AE}\bra{\psi}_{AE}]$ with partial trace $\Tr_E$ over system $E$. As system $E$ purifies system $A$, it already holds the maximal information about system $A$ and hence state $\rho_A$.
As the coherence of $\rho_A$ is defined on the local basis $\mathbb{I}_A=\{\ket{i_A}\}_{i_A=0,\cdots, d_A-1}$, we consider a measurement on system $A$ on the basis $\mathbb{I}_A$. Equivalently, we can consider a dephasing channel and describe the state after the measurement on $\mathbb{I}_A$ as
\begin{equation}\label{Eq:rhoxe}
\begin{aligned}
	\rho_{X_{A}E} &= \Delta_A(\ket{\psi}_{AE}\bra{\psi}_{AE})\\
	& = \sum_{i}(\ket{i_A}\bra{i_A}\otimes I_E)\ket{\psi}_{AE}\bra{\psi}_{AE}(\ket{i_A}\bra{i_A}\otimes I_E).
\end{aligned}
\end{equation}
Here $I_E$ is the identity matrix on system $E$, $\rho_{X_{A}E}$ is a classical-quantum state, and $X_A$ denotes the system of classical outcomes for the measurement.
In the following, we will relate the six coherence measures to a general conditional entropy of the state $\rho_{X_{A}E}$.

First, we consider the relative entropy of coherence $C_r(\rho_{A})$ and the Von-Neumann conditional entropy $H(A|B)_{\rho_{AB}} = S(\rho_{AB}) - S(\rho_B)$. According to Ref.~\cite{coles2012unification,yuan2016interplay}, we can re-express $C_r(\rho_{A})$ as a conditional entropy $H(X_A|E)_{\rho_{X_{A}E}}$,
\begin{equation}\label{crbyh}
  C_r(\rho_A)=H(X_A|E)_{\rho_{X_{A}E}}.
\end{equation}
The operational meaning {of the above definition} is that $H(X_A|E)_{\rho_{X_{A}E}}$  measures the randomness of the measurement outcome $X_A$ conditioned on system $E$. As system $E$ holds the purification of system $A$ before the measurement, $C_r(\rho_A)$ thus describes the unpredictable randomness conditioned on all other quantum systems.

Now, we focus on $C_{\min}(\rho)$ and $C_{\max}(\rho)$ and the conditional min and max entropy defined as
\begin{equation}\label{Hminmax}
\begin{aligned}
H_{\min}(A|B)_{\rho_{AB}}&=\max\{\lambda\in\mathbb{R}|\exists\sigma_B\in\mathcal{D}(\mathcal{H}_B):\rho_{AB}\leq 2^{-\lambda}I_A\otimes\sigma_B\},\\
H_{\max}(A|B)_{\rho_{AB}}&=\max_{\sigma_B}\log_2 F(I_A\otimes \sigma_B,\rho_{AB}).
\end{aligned}
\end{equation}
We show that $C_{\min}(\rho)$, $C_{\max}(\rho)$ and  $C_g(\rho)$  can be defined by the conditional min and max entropy of $\rho_{X_{A}E}$.
\begin{theorem}\label{theorem:main}
  Let $\rho_A$ be a quantum state in $\mathcal{D}(\mathcal{H}_A)$ and let $\rho_{X_{A}E}$ be defined in Eq.~\eqref{Eq:rhoxe}, then we have
  \begin{equation}\label{coherenceeqcond}
  \begin{aligned}
  C_{\max}(\rho_A)&=H_{\max}(X_A|E)_{\rho_{X_{A}E}},\\
  C_{\min}(\rho_A)&=H_{\min}(X_A|E)_{\rho_{X_{A}E}},\\
  C_g(\rho_{A})&=1-2^{-H_{\min}(X_A|E)_{\rho_{X_{A}E}}}.
  \end{aligned}
  \end{equation}
\end{theorem}

To prove the result, we first introduce the isometry
\begin{equation}
	V=\sum_{i}\ket{i_A}_{X_A}\otimes\ket{i_A}\bra{i_A}\otimes I_E,
\end{equation}
 which maps systems $AE$ to $X_AAE$. Under this map, one can equivalently understand $X_A$ as the measurement outcome and  $A$ as the state after measurement. Then, the dephasing operator can be equivalently expressed as
\begin{equation}
	\Delta_A(\ket{\psi}_{AE}\bra{\psi}_{AE}) = \Tr_A\left[V\ket{\psi}_{AE}\bra{\psi}_{AE}V^{\dagger}\right].
\end{equation}
With this new definition, we first prove the following lemma.

  \begin{lemma}\label{isometry}
\begin{equation}\label{lemmaisometry}
  D_{\max}(\rho||\sigma)=D_{\max}(V\rho V^\dagger||V\sigma V^\dagger)
\end{equation}
for any isometry  $V$ that satisfies $V^\dagger V=I$.
\end{lemma}
\begin{proof}
  First of all, we prove that $\rho\geq 0$ if and only if $V\rho V^\dagger\geq 0$. The only if part is obvious. For the if part, $\forall \ket{\phi}$, since $V\rho V^\dagger\geq 0$ we have $\bra{\phi}V^\dagger(V\rho V^\dagger)V\ket{\phi}\geq 0$, therefore $\bra{\phi}\rho\ket{\phi}\geq 0$ and we have $\rho\geq 0$.
  \par
  Since
  \begin{equation}\label{dmaxproof}
  \begin{aligned}
    \lambda\sigma-\rho\geq 0&\Leftrightarrow V(\lambda\sigma-\rho)V^\dagger\geq 0\\
    &\Leftrightarrow\lambda V\sigma V^\dagger-V\rho V^\dagger\geq 0,
  \end{aligned}
  \end{equation}
  we conclude that $D_{\max}(\rho||\sigma)=D_{\max}(V\rho V^\dagger||V\sigma V^\dagger)$.
\end{proof}
Notice that the fidelity is also invariant under isometric channel, i.e.
\begin{equation}\label{fidelityisometry}
  F(\rho,\sigma)=F(V\rho V^\dagger,V\sigma V^\dagger)
\end{equation}
where $V^\dagger V=I$.
Now, we prove Theorem \ref{theorem:main}.
\begin{proof}
Suppose $\ket{\psi}_{AE}$ is a pure state such that $\rho_A=\Tr_E[\ket{\psi}_{AE}\bra{\psi}_{AE}]$, then $\rho_{X_AAE}:=V\ket{\psi}_{AE}\bra{\psi}_{AE}V^{\dagger}$ is also a pure state. From the duality between the conditional min and max entropy \cite{konig2009operational}, we have
\begin{equation}\label{duality}
  \begin{aligned}
  &H_{\min}(X_A\vert E)=-H_{\max}(X_A\vert A),\\
  &H_{\max}(X_A\vert E)=-H_{\min}(X_A\vert A).
  \end{aligned}
\end{equation}
We omit the subscript state since all systems are taken from the state $\rho_{X_AAE}$.
Let $M=\sum_{j}\ket{j}_{X_A}\otimes\Pi_j$ where $\Pi_j=\ket{j}\bra{j}_A$, then $\rho_{X_AA}=M\rho_A M^\dagger$ and we have
\begin{equation}\label{hminproof}
  \begin{aligned}
  H_{\min}(X_A\vert E)&=-H_{\max}(X_A\vert A)\\
  &=-\max_{\sigma\in\mathcal{D}(\mathcal{H}_A)}\log_2 F(\rho_{X_AA},I_{X_A}\otimes\sigma)\\
  &=-\max_{\sigma\in\mathcal{D}(\mathcal{H}_A)}\log_2 F(M\rho_A M^\dagger,MM^\dagger I_{X_A}\otimes\sigma MM^\dagger)\\
  &=-\max_{\sigma\in\mathcal{D}(\mathcal{H}_A)}\log_2 F(\rho_A,M^\dagger I_{X_A}\otimes\sigma M)\\
  &=-\max_{\sigma\in\mathcal{D}(\mathcal{H}_A)}\log_2 F\left(\rho_A,\sum_j\Pi_j\sigma\Pi_j\right)\\
  &=-\max_{\delta\in \mathcal{I}}\log_2 F(\rho_A,\delta)\\
  &=\min_{\delta\in\mathcal{I}}D_{\min}(\rho_A||\delta)\\
  &=C_{\min}(\rho_A).
  \end{aligned}
\end{equation}
The third line uses the fact that $F(\rho,\sigma)=F(\rho,\Pi_\rho\sigma\Pi_\rho)$ and the fourth line uses Eq. \eqref{fidelityisometry}. Note that, a similar result of Eq.~\eqref{hminproof} can also be found in \cite{coles2012unification}.

Similarly, we also prove the relation between $C_{\max}$ and conditional max entropy.
\begin{equation}\label{hmaxproof}
  \begin{aligned}
  H_{\max}(X_A\vert E)&=-H_{\min}(X_A\vert A)\\
  &=\min_{\sigma\in\mathcal{D}(\mathcal{H}_A)}D_{\max}(\rho_{X_AA}||I_{X_A}\otimes\sigma)\\
  &=\min_{\sigma\in\mathcal{D}(\mathcal{H}_A)}D_{\max}(M\rho_A M^\dagger||MM^\dagger I_{X_A}\otimes\sigma MM^\dagger)\\
  &=\min_{\sigma\in\mathcal{D}(\mathcal{H}_A)}D_{\max}(\rho_A||M^\dagger I_{X_A}\otimes\sigma M)\\
  &=\min_{\delta\in\mathcal{I}}D_{\max}(\rho_A||\delta)\\
  &=C_{\max}(\rho_A).
  \end{aligned}
\end{equation}
Here, the third line is because $\rho_{X_AA}$ is diagonal and the fourth line follows from Lemma \ref{isometry}. Finally, the proof for $C_g$ follows directly from its relation with $C_{\min}$.
\end{proof}

\subsection{Conditioned on classical information}
Now, we consider the coherence measures $C_f(\rho_A)$ and $C_0(\rho_A)$. In the above analysis, the conditional entropy is conditioned on the quantum information of system $E$. Alternatively, we can also consider the case where system $E$ performs a measurement $M_E$ and obtain an outcome $X_E$. Suppose $M_E$ is a projective measurement on basis $\mathbb{J}_E=\{\ket{j_E}\}_{j_E=0,\cdots, d_E-1}$, then we can similarly define the measurement or the dephasing operator on system $E$ as $\Delta_E(\rho_{AE})=\sum_{j}(I_A\otimes\ket{j_E}\bra{j_E})\rho_{AE}(I_A \otimes \ket{j_E}\bra{j_E})$. The state after the measurements on $\mathbb{I}_A$ and $\mathbb{J}_E$ becomes
\begin{equation}\label{Eq:rhoxaxe}
\begin{aligned}
	\rho_{X_AX_E} &= \Delta_A\left(\Delta_E(\ket{\psi}_{AE}\bra{\psi}_{AE})\right)\\
	& = \sum_{i,j}(\ket{i_A}\bra{i_A}\otimes \ket{j_E}\bra{j_E})\ket{\psi}_{AE}\bra{\psi}_{AE}(\ket{i_A}\bra{i_A}\otimes \ket{j_E}\bra{j_E}),
\end{aligned}
\end{equation}
where the system $X_E$ denotes the outcomes of measurement $M_E$.
Then, we can use the conditional entropy of the classical-classical state $\rho_{X_AX_E}$ to determine the coherence of $\rho_A$. Since $\rho_{X_AX_E}$ depends on the measurement $M_E$ of system $E$, we thus also consider a minimization over all measurements $M_E$. Especially, when using the conditional Von-Neumann entropy, the coherence of formation can be given by \cite{yuan2015intrinsic},
\begin{equation}\label{Eq:formation}
  C_f(\rho_A)=\min_{M_E}H(X_A|X_E)_{\rho_{X_AX_E}},
\end{equation}
where $X_A$ and $X_E$ denote the classical outcome of Alice and Eve's measurement on joint purification state $\ket{\psi}_{AE}$, respectively, and the minimization is over all measurements on system $E$.
Here, we restate the proof of Eq.~\eqref{Eq:formation} {in a more straightforward way than Ref.~\cite{yuan2015intrinsic}}.
\begin{theorem}\label{theorem:formation}
  \begin{equation*}
  C_f(\rho_A)=\min_{M_E}H(X_A|X_E)_{\rho_{X_AX_E}}.
\end{equation*}
\end{theorem}
\begin{proof}
Suppose the initial state shared by A and E is $\ket{\psi}_{AE}$. After performing measurement $M_E$ on system $E$, which is a projective measurement onto the basis $\mathbb{J}_E=\{\ket{j_E}\}$, the state becomes
\begin{equation}
  \begin{aligned}
\rho_{AX_E}=\sum_j p_j \ket{\psi_j}\bra{\psi_j}_A \otimes \ket{j_E}\bra{j_E},
\end{aligned}
\end{equation}
where $\{p_j,\ket{\psi_j}\}$ is an ensemble of state $\rho_A$. Let $\ket{\psi_j}=\sum_i a_{ji} \ket{i_A}$. Then we perform measurement on the computational basis $\mathbb{I}_A=\{\ket{i_A}\}$ which results in the state
\begin{equation}
  \begin{aligned}
\rho_{X_AX_E}=\sum_j  \sum_i p_j|a_{ji}|^2 \ket{i_A}\bra{i_A} \otimes \ket{j_E}\bra{j_E}.
\end{aligned}
\end{equation}
By definition,
\begin{equation}\label{Cfdeduction}
\begin{aligned}
H(X_A|X_E)_{\rho_{X_AX_E}}&=S(\rho_{X_AX_E})-S(\rho_{X_E})\\
&=\sum_{j,i}p_j|a_{ji}|^2\log_2\frac{1}{p_j|a_{ji}|^2}-\sum_jp_j\log_2\frac{1}{p_j}\\
&=\sum_jp_jS(\Delta(\ket{\psi_j}\bra{\psi_j})).
\end{aligned}
\end{equation}
Since a different measurement on system $E$ corresponds to different convex roof decomposition of state $\rho_A$, we conclude that
\begin{equation}\label{cffinal}
\begin{aligned}
  C_f(\rho_A)&=\min_{\{p_j,\ket{\psi_j}\}}\sum_jp_jS(\Delta(\ket{\psi_j}\bra{\psi_j}))\\
  &=\min_{\{p_j,\ket{\psi_j}\}}H(X_A|X_E)_{\rho_{X_AX_E}}\\
  &=\min_{M_E}H(X_A|X_E)_{\rho_{X_AX_E}}.
  \end{aligned}
\end{equation}

\end{proof}

For the coherence measure $C_0$, we also have a similar result with the conditional 0-entropy
\begin{equation}
	H_0(A|B)_{\rho_{AB}}=\max_{\sigma_B}\log_2\Tr\left[\Pi_{\rho_{AB}}(I_A\otimes\sigma_B)\right].
\end{equation}
Here $\Pi_\rho$ denotes projection onto the support of $\rho$.

\begin{theorem}\label{theorem0entropy}
	\begin{equation}\label{c0}
  C_0(\rho)=\min_{M_E}H_0(X_A|X_E)_{\rho_{X_AX_E}}.
\end{equation}
\end{theorem}

\begin{proof}
Similar to the proof of Theorem \ref{theorem:formation}, we perform a measurement $M_E$ on system $E$ which leads to a state
\begin{equation}
  \begin{aligned}
\rho_{AX_E}=\sum_j p_j \ket{\psi_j}\bra{\psi_j}_A \otimes \ket{j_E}\bra{j_E},
\end{aligned}
\end{equation}
where $\{p_j,\ket{\psi_j}\}$ is an ensemble of state $\rho_A$. Let $\ket{\psi_j}=\sum_i a_{ji} \ket{i_A}$. Then we perform measurement on the computational basis $\mathbb{I}_A=\{\ket{i_A}\}$ which results in the classical-classical state
\begin{equation}
  \begin{aligned}
\rho_{X_AX_E}=\sum_j  \sum_i p_j|a_{ji}|^2 \ket{i_A}\bra{i_A} \otimes \ket{j_E}\bra{j_E}.
\end{aligned}
\end{equation}

According to the definition of conditional 0-entropy,

\begin{equation}
  \begin{aligned}
  \min_{M_E} H_0(A|E)_{\rho_{X_AX_E}}&=\min_{M_E} \max_{\sigma_E }\{-D_0(\rho_{X_AX_E}||I_A\otimes \sigma_E)\}\\
&=\min_{M_E} \max_{\sigma_E } \log_2 \mathrm{Tr}[\Pi_{\rho_{X_AX_E}} (I_A\otimes\sigma_E))]
\end{aligned}
\end{equation}
where $\Pi_{\rho_{X_AX_E}}$ is the projection onto the support of $\rho_{X_AX_E}$ which can be written as
\begin{equation}
\Pi_{\rho_{X_AX_E}}=\sum_{j,i}\mathds{1}[a_{ji}\neq 0] \ket{i_A}\bra{i_A} \otimes \ket{j_E}\bra{j_E}
\end{equation}
where $\mathds{1}[X]=1$ if $X$ is true, and 0 otherwise.
Any state $\sigma_E\in \mathcal{D}(\mathcal{H}_E)$ can be expressed in the basis $\{\ket{j_E}\}$ as
\begin{equation}
  \begin{aligned}
\sigma_E=\sum_{j,j'} b_{jj'} \ket{j_E}\bra{j_E'}.
\end{aligned}
\end{equation}
Thus

\begin{equation}
  \begin{aligned}
&\mathrm{Tr}[\Pi_{\rho_{X_AX_E}} (I_A\otimes\sigma_E)]\\
&=\mathrm{Tr}\left[\sum_{j,j',i}b_{jj'}\mathds{1}[a_{ji}\neq 0] \ket{i_A}\bra{i_A} \otimes \ket{j_E}\bra{j_E'}\right]\\
&=\sum_j b_{jj} T_j,
\end{aligned}
\end{equation}
where $T_j=\sum_i\mathds{1}[a_{ji}\neq 0]$.
Then we know that $\max_{\sigma_E } \log_2 \mathrm{Tr}(\Pi_{\rho_{AE}} (I_A\otimes\sigma_E))=\max_j\log_2 T_j$. Also, since the measurement chosen by system E corresponds to an ensemble of state $\rho_A$, we have
\begin{equation}\label{C0deduction}
  \begin{aligned}
  C_{0}(\rho)&=\min_{\{p_j,\ket{\psi_j}\}} \max_j \log_2 T_j\\
  &=\min_{M_E} \max_j \log_2 T_j\\
  &=\min_{M_E} H_0(A|E)_{\rho_{X_AX_E}}.
  \end{aligned}
\end{equation}
\end{proof}




\section{Computing the coherence measure via SDP}\label{Sec:SDP}

By definition, we can see that the relative entropy of coherence $C_r(\rho)$ can be easily computed. However, none of the other five coherence measures $C_g$, $C_{\max}$, $C_{\min}$,  $C_f$, and $C_0$ can be directly computed from the definition as they all involve a minimization procedure.  In literature, an upper and lower bound of the geometric coherence measure is proposed in Ref.~\cite{zhang2017estimation}.  In this work, we focus on numerical calculations of $C_g$, $C_{\max}$, $C_{\min}$ based on the unified definition in the last section. As $C_f$ and $C_0$ are defined {by conditioning on classical information, whether they can be efficiently computed} is still left as an open problem.

The main reason that we can efficiently compute $C_g$, $C_{\max}$, $C_{\min}$ is based on the fact that the conditional max/min entropy $H_{\min}(A|B)_{\rho_{AB}}$ and $H_{\max}(A|B)_{\rho_{AB}}$  can be efficiently computed via semi-definite programming (SDP) \cite{tomamichel2012framework}, that is,
\begin{equation}\label{hminsdp}
\begin{aligned}
  H_{\min}(A|B)_{\rho_{AB}} &= -\log_2OPT\\
  \text{where } OPT &=  \min \Tr[\sigma_B]\\
  \text{s.t. }& I_A\otimes \sigma_B\geq \rho_{AB}\\
  & \sigma_B\geq 0
\end{aligned}
\end{equation}

\begin{equation}\label{hmaxsdp}
\begin{aligned}
  H_{\max}(A|B)_{\rho_{AB}} &= \log_2 OPT\\
  \text{where }OPT &= \min\mu\\
  \text{s.t. }& \mu I_B\geq \Tr_A[\sigma_{AB}]\\
  &\sigma_{AB}\otimes I_C\geq \rho_{ABC}\\
  &\sigma_{AB}\geq 0\\
  & \mu\geq 0
\end{aligned}
\end{equation}
In Eq. \eqref{hmaxsdp}, $\rho_{ABC}$ is an arbitrary purification of $\rho_{AB}$.

Our result, Theorem 3, bridges the coherence measures and the conditional max/min entropies. The numerical method to compute $C_{\min}(\rho)$ and $C_{\max}(\rho)$ by SDP is described as follows.

\begin{equation}\label{cminsdp}
\begin{aligned}
  C_{\min}(\rho) &= H_{\min}(X_A|E)_{\rho_{X_{A}E}}\\
  \text{where }&\rho_{AE}=\mathrm{Purification}(\rho)\\
  & \rho_{X_{A}E} = \Tr_A[V\rho_{AE}V^{\dagger}]\\
  & V=\sum_{i}\ket{i_A}_{X_A}\otimes\ket{i_A}\bra{i_A}\otimes I_E
\end{aligned}
\end{equation}

\begin{equation}\label{cmaxsdp}
\begin{aligned}
  C_{\max}(\rho) &= H_{\max}(X_A|E)_{\rho_{X_{A}E}}\\
  \text{where }&\rho_{AE}=\mathrm{Purification}(\rho)\\
  & \rho_{X_{A}E} = \Tr_A[V\rho_{AE}V^{\dagger}]\\
  & V=\sum_{i}\ket{i_A}_{X_A}\otimes\ket{i_A}\bra{i_A}\otimes I_E
\end{aligned}
\end{equation}
Note that the geometric coherence is related to the min-entropy of coherence as $C_g(\rho)=1-2^{-C_{\min}(\rho)}$, so it can be computed via Eq. \eqref{cminsdp} as well.

Here we show several examples of the calculation of different coherence measures. First, we consider a type of symmetric mixed state in qubit case and high dimension case, $\rho=\nu\ket{+}\bra{+}+(1-\nu)\frac{I}{2}$ and $\rho=\nu(\ket{+}\bra{+})^{\otimes 3}+(1-\nu)\frac{I^{\otimes 3}}{8}$, respectively. The comparisons of different coherence measures, $C_{\max}(\rho),C_r(\rho),C_{\min}(\rho),C_g(\rho)$, are shown in
Fig.~\ref{example1} and Fig.~\ref{example2}. The order of the coherence measures as proven in Theorem~\ref{relationship} is also clearly shown.

\begin{figure}[t]
\centering
{\resizebox{10cm}{!}{\includegraphics{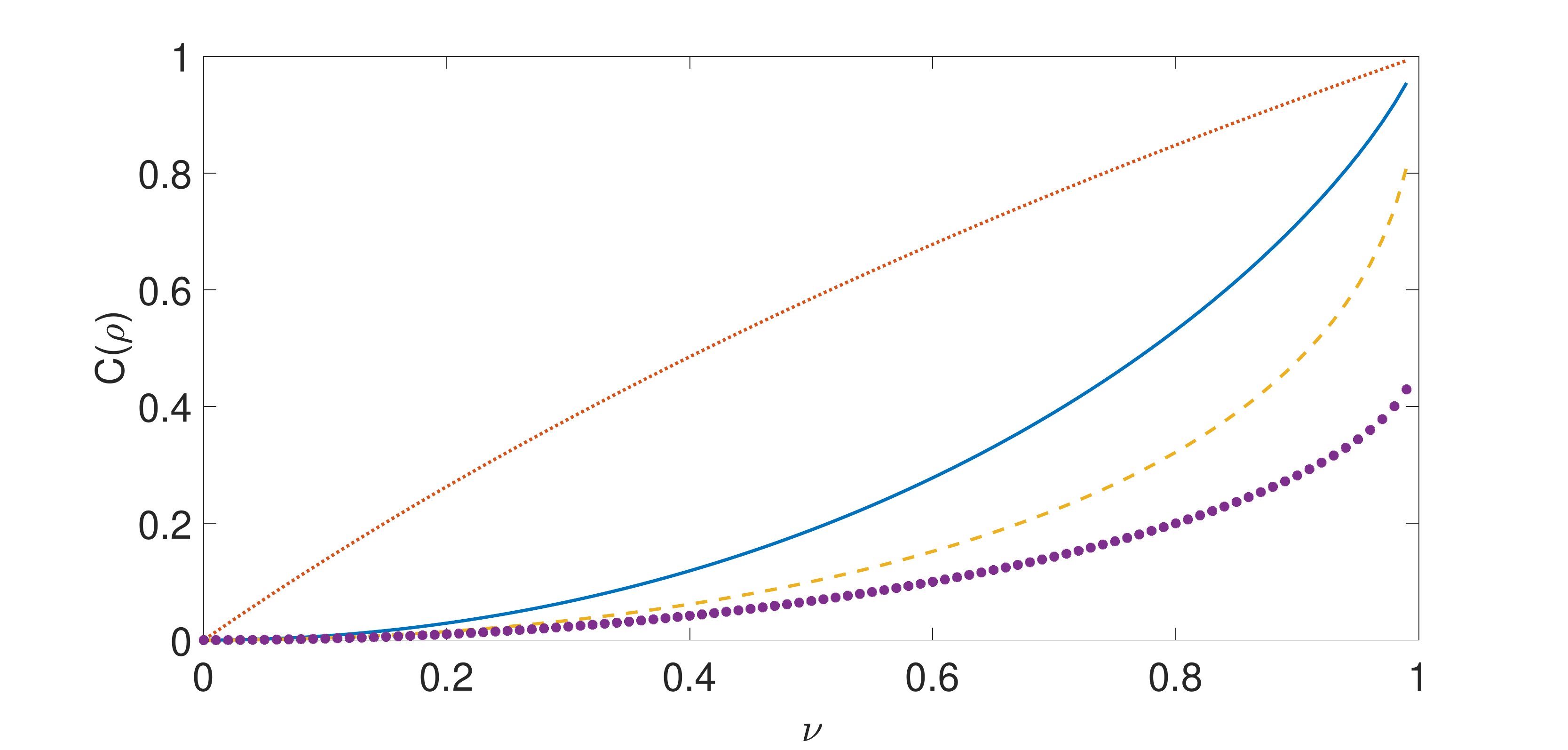}}}
\caption{Coherence of the state $\rho=\nu\ket{+}\bra{+}+(1-\nu)\frac{I}{2}$. From above to below, the lines correspond to $C_{\max}(\rho),C_r(\rho),C_{\min}(\rho),C_g(\rho)$, respectively.}\label{example1}
\end{figure}

\begin{figure}[t]
\centering
{\resizebox{10cm}{!}{\includegraphics{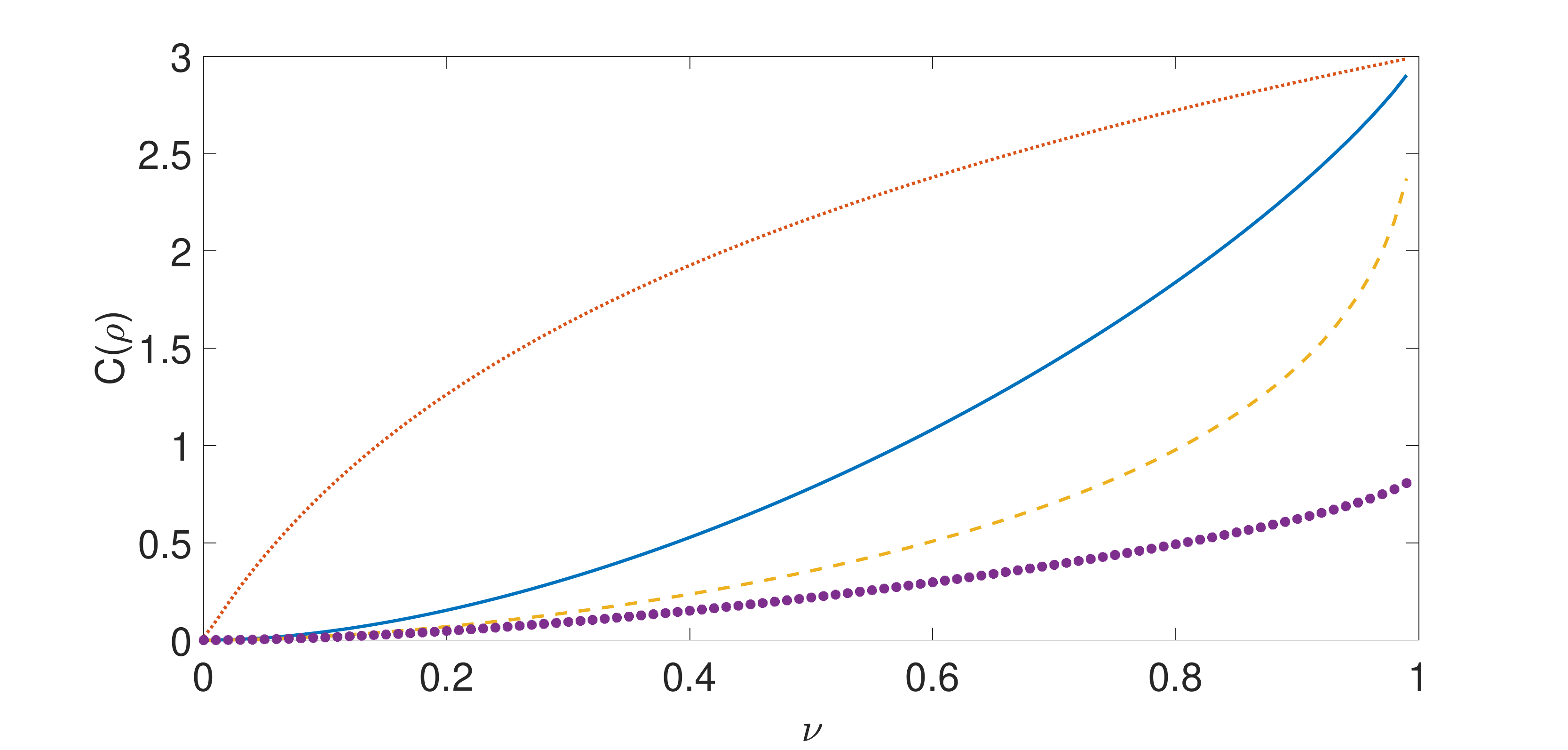}}}
\caption{Coherence of the state $\rho=\nu(\ket{+}\bra{+})^{\otimes 3}+(1-\nu)\frac{I^{\otimes 3}}{8}$. From above to below, the lines correspond to $C_{\max}(\rho),C_r(\rho),C_{\min}(\rho),C_g(\rho)$, respectively.}
\label{example2}
\end{figure}

Next we plot the coherence measures for all qubit states. A qubit $\rho=\frac{I+\vec{n}\cdot\vec{\sigma}}{2}$ is specified by its coordinate $\vec{n}=(n_x,n_y,n_z)$ on the Bloch sphere where $\vec{\sigma}=(X,Y,Z)$ are Pauli matrices. We can also describe the coordinate by longitude $\phi$, latitude $\theta$, and distance to the z-axis $r$. However, since coherence is defined on z-basis, the x and y basis are symmetric, meaning that two states have the same amount of coherence if only their longitude are different. So we can use two parameters $\beta,\gamma$ to represent a state where $\theta = \beta\pi$ is the latitude and $\gamma=\frac{r}{\sin\theta}$ is the normalized distance to the z-axis. Here both parameters $\beta,\gamma$ are in the range $[0,1]$. We let the longitude $\phi=0$, then a qubit is specified by $\beta,\gamma$ as
\begin{equation}\label{representation}
  \vec{n}=\left(\gamma\sin(\beta\pi),0,\cos(\beta\pi)\right).
\end{equation}
The results are presented in Fig. \ref{allqubit} where we can also see and compare the four coherence measures.

\begin{figure}[t]
\subfigure[]{\includegraphics[height=4.2cm]{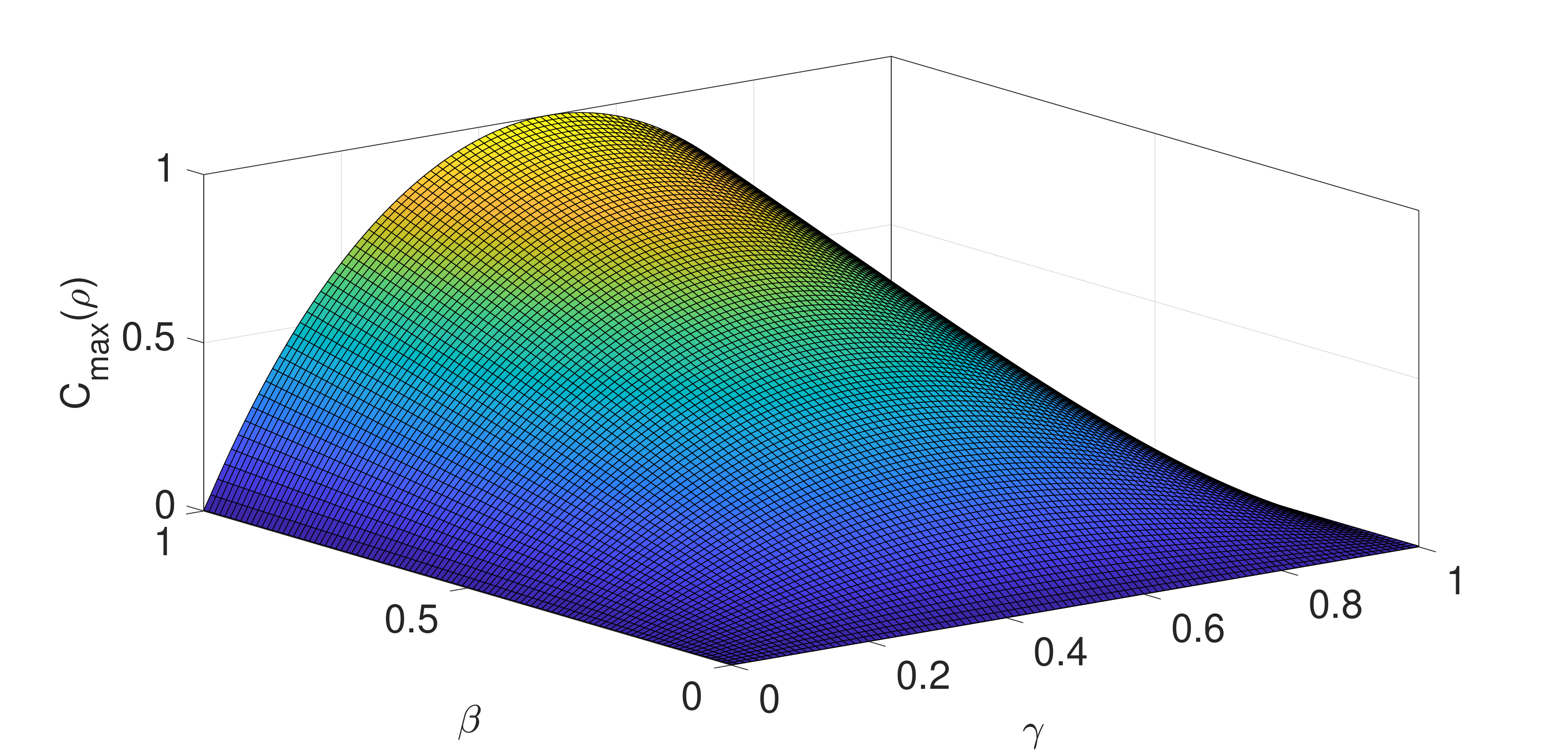}}
\subfigure[]{\includegraphics[height=4.2cm]{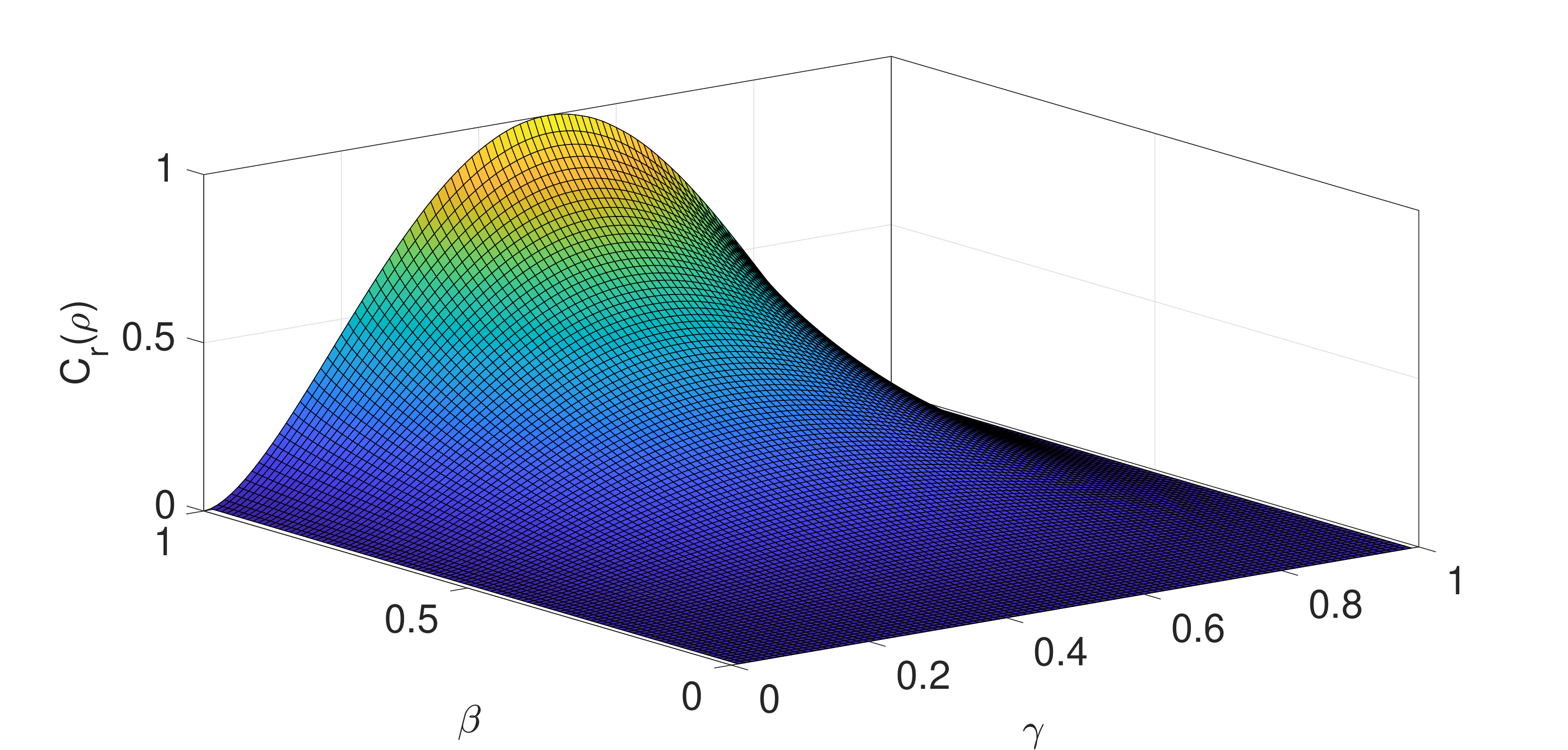}}
\subfigure[]{\includegraphics[height=4.2cm]{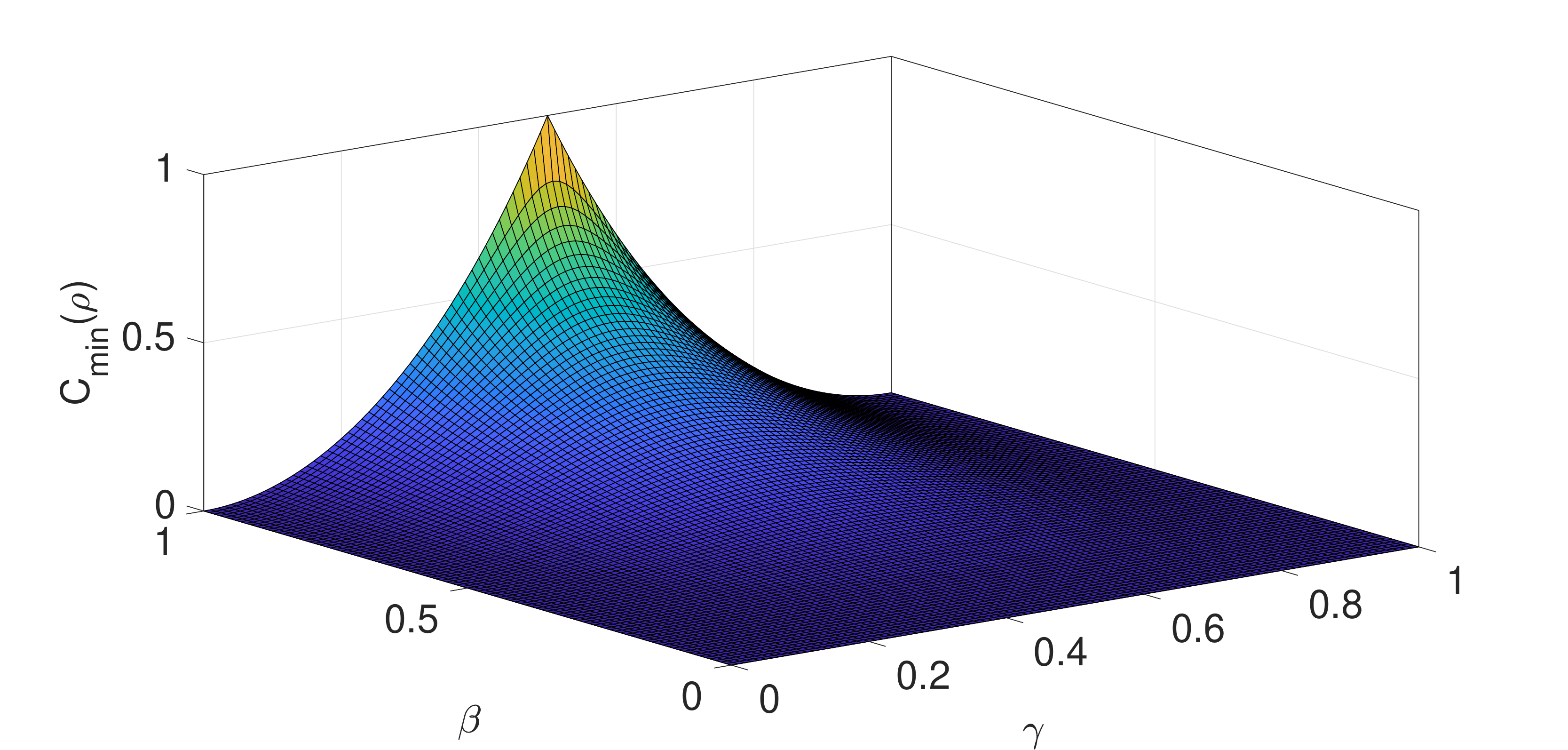}}
\subfigure[]{\includegraphics[height=4.2cm]{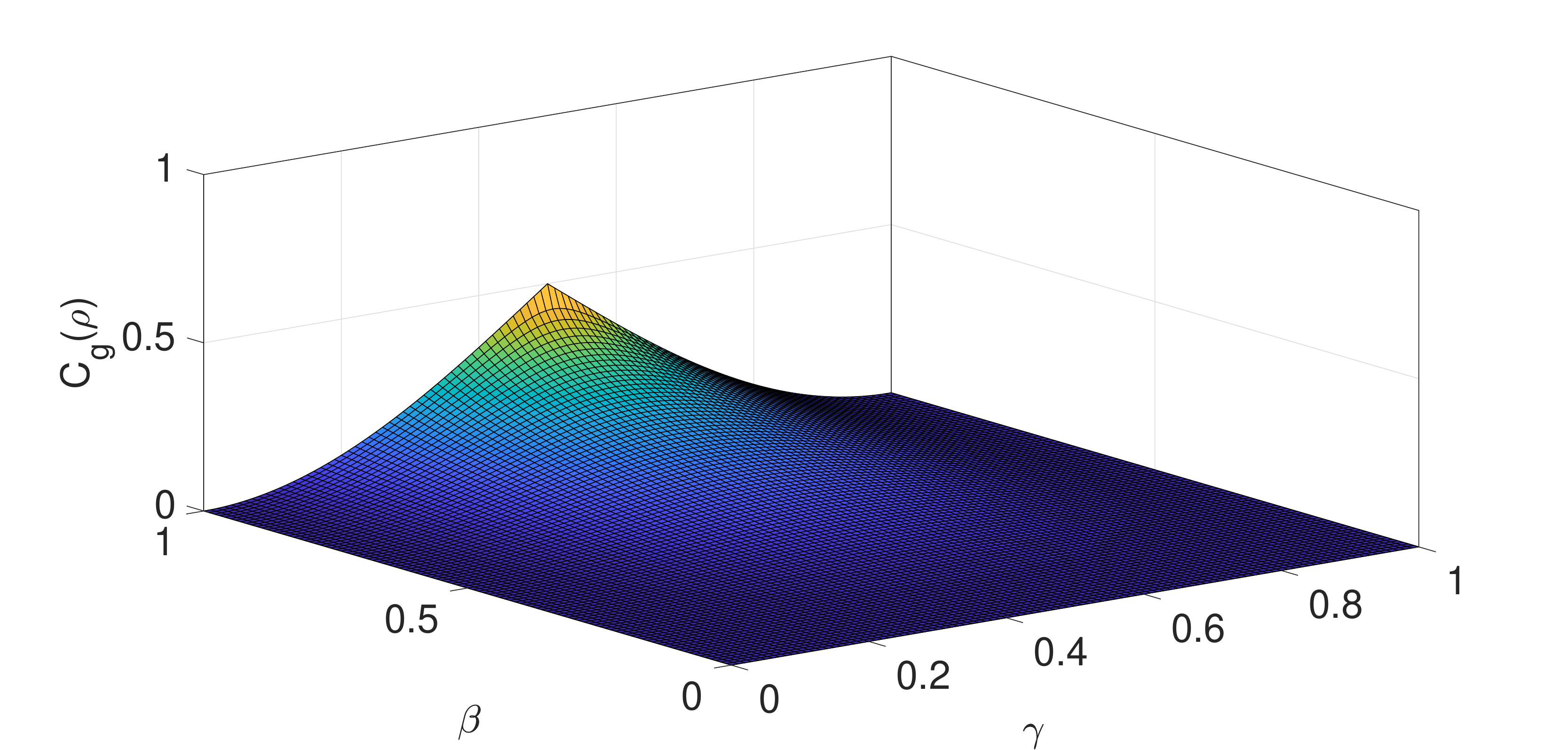}}
\caption{Coherence measures (a) $C_{\max}(\rho)$, (b) $C_{r}(\rho)$, (c) $C_{\min}(\rho)$, and (d) $C_{g}(\rho)$ for all qubit states.}\label{allqubit}
\end{figure}

\section{Operational meaning}\label{Sec:operation}

In this section, we study the operational meaning of the coherence measures that are defined in the unified way via general conditional entropies.
The quantum coherence of $\rho_A$ characterizes the unpredictable randomness of the state. That is, considering any adversary that is entangled with $\rho_A$, she has certain uncertainty, quantified by coherence measures, about the internal property of $\rho_A$ in the coherence basis. This interpretation can be understood from the quantum conditional entropies, where the conditional entropy of $A$ conditioned on $E$ characterizes some kind of uncertainty of $E$ with respective to the information of the system $A$. Here we show that beyond the conventional conditional Von-Neumann entropy, some of the generalized conditional entropies on a classical-quantum state defined in Eq.~\eqref{Eq:rhoxe} also have strong operational meanings.

{Rewrite Eq. \eqref{Eq:rhoxe} as follows,
\begin{equation}\label{Eq:rhoxecqstate}
\begin{aligned}
  \rho_{X_{A}E} &= \sum_{i}(\ket{i_A}\bra{i_A}\otimes I_E)\ket{\psi}_{AE}\bra{\psi}_{AE}(\ket{i_A}\bra{i_A}\otimes I_E)\\
    & = \sum_{i}p_i\ket{i_A}\bra{i_A}\otimes\rho_E^i,
\end{aligned}
\end{equation}
where $p_i=\Tr\left[\ket{\psi}_{AE}\bra{\psi}_{AE}(\ket{i_A}\bra{i_A}\otimes I_E)\right]$ and $\rho_E^i=(\bra{i_A}\otimes I_E)\ket{\psi}_{AE}\bra{\psi}_{AE}(\ket{i_A}\otimes I_E)/p_i$.}
On the one hand, as described in Ref.~\cite{konig2009operational}, the conditional min-entropy $H_{\min}(X_A\vert E)_{\rho_{X_{A}E}}$ of a classical-quantum state $\rho_{X_{A}E}$ corresponds to the guessing probability $p_{guess}(X_A\vert E)$,
\begin{equation}
H_{\min}(X_A\vert E)_{\rho_{X_{A}E}}=-\log_2 p_{\textrm{guess}}(X_A\vert E).
\end{equation}
Here $p_{\textrm{guess}}(X_A\vert E)$ is the maximum probability that Eve could guess $X_A$ correctly according to her system with an optimal measurement $M_E=\{E_i\}$,
\begin{equation}
p_{\textrm{guess}}(X_A\vert E)=\max_{\{E_i\}}\sum_i p_i\Tr[E_i \rho_E^i].
\end{equation}
As $C_{\min}(\rho) = H_{\min}(X_A\vert E)_{\rho_{X_{A}E}}$ with $\rho_{X_{A}E}$ being the dephased state defined in Eq.~\eqref{Eq:rhoxe}, we thus derive an operational meaning of the coherence measure $C_{\min}(\rho)$. Note that a similar argument is also discovered in  \cite{coles2012unification}.

On the other hand, the conditional max-entropy $H_{\max}(X_A\vert E)_{\rho_{X_{A}E}}$ of a classical-quantum state $\rho_{X_{A}E}$ corresponds to the security of $X_A$ when used as a secret key in the presence of adversary $E$,
\begin{equation}
H_{\max}(X_A\vert E)_{\rho_{X_{A}E}}=\log_2 p_{\textrm{secr}}(X_A\vert E).
\end{equation}
Here the security of the secret key can be quantified by the maximum fidelity between $\rho_{X_{A}E}$ and $I_{X_A}/|X_A|\otimes \rho_E$ where $I_{X_A}/|X_A|$ is the maximally mixed state which corresponds to the uniform distribution and $|X_A|$ is the alphabet size of $X_{A}$,
\begin{equation}
p_{\textrm{secr}}(X_A\vert E)=|X_A| \max_{\sigma_E}F(\rho_{X_{A}E},I_{X_A}/|X_A|\otimes \sigma_E).
\end{equation}
Here the state  $I_{X_A}/|X_A|\otimes \sigma_E$ means that the key is uniformly distributed and independent of Eve's system. As $H_{\max}(X_A\vert E)_{\rho_{X_{A}E}}$ is related to $C_{\max}(\rho)$, we thus give an operational meaning to $C_{\max}(\rho)$.


The operational meaning of the coherence measures $C_{\min}(\rho)$ and $C_{\max}(\rho)$  are both built on that the adversary Eve is a \emph{quantum adversary} which possess a quantum system that shares entanglement with Alice's state. In general, coherence measures that belong to this framework can be expressed as
\begin{equation}\label{quantumadversary}
  C_{\alpha}(\rho_A)=H_{\alpha}(X_A|E)_{\rho_{X_AE}}
\end{equation}
where $H_{\alpha}$ denotes general conditional entropy and $C_{\alpha}$ denotes the corresponding coherence measure. We expect similar results can be found for other $\alpha$.

Instead, we can also consider Eve to be a \emph{classical adversary} and build the corresponding framework of conditional entropy with Eve's system being a classical random variable. In this case, we require that Eve  performs a measurement $M_E$ {on her local basis $\{\ket{j_E}\}$} and infers Alice's information from her own measurement result. {Following Eq. \eqref{Eq:rhoxaxe}, after both party's measurement, the resulting state is a classical-classical state 
\begin{equation}
  \rho_{X_AX_E}=\sum_j  \sum_i p_j|a_{ji}|^2 \ket{i_A}\bra{i_A} \otimes \ket{j_E}\bra{j_E},
\end{equation}
where $p_j=\Tr\left[\ket{\psi}_{AE}\bra{\psi}_{AE}(I_A\otimes \ket{j_E}\bra{j_E})\right]$ and $a_{ji} = (\bra{i_A}\otimes\bra{j_E})\ket{\psi}_{AE}/\sqrt{p_j}$.} In the classical adversary framework, coherence measures are expressed as
\begin{equation}\label{classicaladversary}
  C_{\alpha}(\rho_A)=\min_{M_E}H_{\alpha}(X_A|X_E)_{\rho_{X_AX_E}}.
\end{equation}
Here the minimization over $M_E$ denotes the optimal strategy of Eve's measurement which also corresponds to the decomposition of Alice's state. This is why coherence measures studied in this work that belong to the classical adversary framework all have the convex roof form, including $C_f(\rho)$ and $C_0(\rho)$. The measure $C_f(\rho)$  is regarded as the intrinsic randomness of the state against classical adversary \cite{yuan2015intrinsic}. Meanwhile $H_{0}(X_A|X_E)_{\rho_{X_AX_E}}$ is related to $C_0(\rho)$ which is a classical version of hypothesis testing \cite{wang2012one}. We also note that, since classical adversary is \emph{weaker} than quantum adversary, the uncertainty of a classical adversary should be larger, meaning that coherence measures with the convex roof form are often larger than the ones with distance form. This is also suggested in Theorem \ref{relationship}. Besides the coherence measures studied in this work, whether other coherence measures possess operational meanings in the conditional entropy framework is a remained and interesting open question.


\section{Discussion}\label{Sec:summary}
In this work we show the connections between coherence measures and generalized quantum conditional entropies. Our result highlights the close relation between single partite coherence and bipartite quantum correlation in a classical-quantum state. {When represented in the conditional entropy form,} some of the coherence measures can be efficiently calculated by semi-definite programming. Meanwhile, besides the operational meanings given by expressing as conditional entropies, some of the coherence measures based on the $\alpha$-Renyi divergence also have strong operational meanings from the perspective of quantum resource theory of coherence, such as coherence dilution and distillation, in both the asymptotic and one-shot cases \cite{Winter16,zhao2017one, 2017arXiv171110512R}. For a general quantum state, the gap between the coherence dilution and distillation ratios can be also calculated by our method. The result will pave the way for studying the reversible properties of coherence resource under different incoherent operations, such as the maximally incoherent operation  \cite{aberg2006quantifying}, the dephasing-covariant incoherent operations \cite{Chitambar16prl,Marvian16}, the incoherent operation \cite{baumgratz2014quantifying} and the strictly incoherent operation \cite{Winter16}.

Besides the examples shown in this work, the relationship between other $\alpha$-Renyi divergence coherence monotones and quantum conditional entropies is still an open question.
Moreover, through the connection of generalized quantum conditional entropies, our work can provide links from the resource theory of coherence to other quantum information processing tasks, such as entanglement cost and distillation \cite{Buscemi2011prlentanglement,Buscemi2010distill}, randomness extraction \cite{ma2017source,hayashi2017secure}, and security analysis of quantum key distribution \cite{Devetak2005distill,renner2008security}. These are directions for future work.

\section{Acknowledgements}
We acknowledge X.~M. for the insightful discussions. This work was supported by the National Natural Science Foundation of China Grants No. 11674193.

All authors contributed equally to this work.

\appendix
\section{Coherence distillation and dilution}
{The asymptotic distillation rate under general incoherent operation is defined as
\begin{equation}
	C_{\textrm{distillation}}^{\infty}(\rho) = \sup R, \textrm{s.t.}, \rho^{\otimes n}\overset{\textrm{O}}{\rightarrow}\overset{\varepsilon}{\approx}\ket{\psi_2}^{\otimes nR} \textrm{ as } n\rightarrow\infty, \varepsilon\rightarrow 0^+.
\end{equation}

The asymptotic dilution rate under general incoherent operation is defined as
\begin{equation}
	C_{\textrm{dilution}}^{\infty}(\rho) = \inf R, \textrm{s.t.}, \ket{\psi_2}^{\otimes nR}\overset{\textrm{O}}{\rightarrow}\overset{\varepsilon}{\approx}\rho^{\otimes n} \textrm{ as } n\rightarrow\infty, \varepsilon\rightarrow 0^+.
\end{equation}

The one-shot distillation rate with $\varepsilon$ error under general incoherent operation is defined as
\begin{equation}
	C_{\textrm{distilation}}^{\varepsilon}(\rho) = \sup \log_2d, \textrm{s.t.}, \rho\overset{\textrm{O}}{\rightarrow}\overset{\varepsilon}{\approx} \ket{\psi_d}.
\end{equation}

The one-shot dilution rate with $\varepsilon$ error under general incoherent operation is defined as
\begin{equation}
	C_{\textrm{dilution}}^{\varepsilon}(\rho) = \inf \log_2d, \textrm{s.t.}, \ket{\psi_d}\overset{\textrm{O}}{\rightarrow}\overset{\varepsilon}{\approx}\rho.
\end{equation}

Here, $\rho\overset{\varepsilon}{\approx}\sigma$ means that $F(\rho,\sigma)\le 1-\varepsilon$, $\ket{\psi_d}=\frac{1}{\sqrt{d}}\sum_{i=0}^{d-1}\ket{i}$ is the (canonical) $d$ dimensional maximally coherent state, and $O$ can be an arbitrary class of incoherent operation, such as IO and MIO. For simplicity, we only give the results for exact one-shot coherence transformation, that is, we define $C_{\textrm{distilation}}^{\textrm{one-shot}}(\rho)=\lim_{\varepsilon\to 0^+}C_{\textrm{distilation}}^{\varepsilon}(\rho)$ and $C_{\textrm{dilution}}^{\textrm{one-shot}}(\rho)=\lim_{\varepsilon\to 0^+}C_{\textrm{dilution}}^{\varepsilon}(\rho)$. Under IO, we have
\begin{equation}
\begin{aligned}
		C_{\textrm{distillation}}^{\infty}(\rho) &= C_r(\rho),\\	C_{\textrm{dilution}}^{\infty}(\rho) &= C_f(\rho), \\
			C_{\textrm{dilution}}^{\textrm{one-shot}}(\rho)&\approx C_{\textrm{min}}(\rho),\\
				C_{\textrm{dilution}}^{\textrm{one-shot}}(\rho)&=C_0(\rho).\\
\end{aligned}
\end{equation}

Under MIO, we have
\begin{equation}
\begin{aligned}
C_{\textrm{distillation}}^{\infty}(\rho) &= C_r(\rho),\\	C_{\textrm{dilution}}^{\infty}(\rho) &= C_r(\rho), \\
C_{\textrm{dilution}}^{\textrm{one-shot}}(\rho)&\approx C_{\textrm{max}}(\rho).\\
\end{aligned}
\end{equation}
We refer to Ref.~\cite{Winter16,zhao2017one} for more information.

\section{Proof of Theorem 1}
Here is the proof of Theorem 1.
\begin{proof}

(C1) Let $C_{\min}(\rho)=D_{\min}(\rho||\delta^\ast)$. Using the property of fidelity, that is, $F(\rho,\sigma)=1$ iff $\rho=\sigma$, we know that
\begin{equation}\label{minc1}
\begin{aligned}
C_{\min}(\rho)=0&\Leftrightarrow \exists \delta\in\mathcal{I},\ \rho=\delta\\
&\Leftrightarrow\rho\in\mathcal{I}.
\end{aligned}
\end{equation}
\par
(C2) To prove the monotonicity property, we notice that  $D_{\min}(\rho||\sigma)$ is non-increasing under quantum channel. Thus, we have
\begin{equation}\label{minc2a}
  \begin{aligned}
  C_{\min}(\rho)&=D_{\min}(\rho||\delta^\ast)\\
  &\geq D_{\min}\left(\sum_nK_n\rho K_n^\dagger\Big|\Big|\sum_nK_n\delta^\ast K_n^\dagger\right)\\
  &\geq\min_{\delta\in\mathcal{I}}D_{\min}\left(\sum_nK_n\rho K_n^\dagger\Big|\Big|\delta\right)\\
  &=C_{\min}\left(\sum_nK_n\rho K_n^\dagger\right).
  \end{aligned}
\end{equation}
\par
(C4) To prove the convexity property, we make use of the joint concavity of the square root of the fidelity,
\begin{equation}\label{concavityfidelity}
  \sqrt{F}\left(\sum_np_n\rho_n,\sum_np_n\sigma_n\right)\geq \sum_np_n\sqrt{F}(\rho_n,\sigma_n).
\end{equation}
Let $C_{\min}(\rho_n)=-\log_2F(\rho_n,\delta_n^\ast)$, then we have
\begin{equation}\label{minc3}
  \begin{aligned}
  \sum_np_nC_{\min}(\rho_n)&=-2\sum_np_n\log_2\sqrt{F}\left(\rho_n,\delta_n^\ast\right)\\
  &\geq -2\log_2\left(\sum_np_n\sqrt{F}\left(\rho_n,\delta_n^\ast\right)\right)\\
  &\geq -2\log_2\left(\sqrt{F}\left(\sum_np_n\rho_n,\sum_np_n\delta_n^\ast\right)\right)\\
  &=D_{\min}\left(\sum_np_n\rho_n\Big|\Big|\sum_np_n\delta_n^\ast\right)\\
  &\geq\min_{\delta\in\mathcal{I}}D_{\min}\left(\sum_np_n\rho_n\Big|\Big|\delta\right)\\
  &=C_{\min}\left(\sum_np_n\rho_n\right).
  \end{aligned}
\end{equation}
\end{proof}

\bibliographystyle{iopart-num}
\bibliography{refsdpcoherence3}

\providecommand{\newblock}{}
\begin{thebibliography}{10}
\expandafter\ifx\csname url\endcsname\relax
  \def\url#1{{\tt #1}}\fi
\expandafter\ifx\csname urlprefix\endcsname\relax\def\urlprefix{URL }\fi
\providecommand{\eprint}[2][]{\url{#2}}

\bibitem{Aharonov67}
Aharonov Y and Susskind L 1967 {\em Phys. Rev.\/} {\bf 155}(5) 1428--1431
  \urlprefix\url{http://link.aps.org/doi/10.1103/PhysRev.155.1428}

\bibitem{Kitaev04}
Kitaev A, Mayers D and Preskill J 2004 {\em Phys. Rev. A\/} {\bf 69}(5) 052326
  \urlprefix\url{http://link.aps.org/doi/10.1103/PhysRevA.69.052326}

\bibitem{Bartlett07}
Bartlett S~D, Rudolph T and Spekkens R~W 2007 {\em Rev. Mod. Phys.\/} {\bf
  79}(2) 555--609
  \urlprefix\url{http://link.aps.org/doi/10.1103/RevModPhys.79.555}

\bibitem{Ma2016QRNG}
Ma X, Yuan X, Cao Z, Qi B and Zhang Z 2016 {\em npj Quantum Information\/} {\bf
  2} 16021 review Article
  \urlprefix\url{http://dx.doi.org/10.1038/npjqi.2016.21}

\bibitem{herrero2017}
Herrero-Collantes M and Garcia-Escartin J~C 2017 {\em Rev. Mod. Phys.\/} {\bf
  89}(1) 015004
  \urlprefix\url{https://link.aps.org/doi/10.1103/RevModPhys.89.015004}

\bibitem{aberg2006quantifying}
{Aberg} J 2006 {\em eprint arXiv:quant-ph/0612146\/} (\textit{Preprint}
  \eprint{quant-ph/0612146})

\bibitem{baumgratz2014quantifying}
Baumgratz T, Cramer M and Plenio M~B 2014 {\em Phys. Rev. Lett.\/} {\bf
  113}(14) 140401
  \urlprefix\url{https://link.aps.org/doi/10.1103/PhysRevLett.113.140401}

\bibitem{yuan2015intrinsic}
Yuan X, Zhou H, Cao Z and Ma X 2015 {\em Phys. Rev. A\/} {\bf 92}(2) 022124
  \urlprefix\url{https://link.aps.org/doi/10.1103/PhysRevA.92.022124}

\bibitem{Napoli16}
Napoli C, Bromley T~R, Cianciaruso M, Piani M, Johnston N and Adesso G 2016
  {\em Phys. Rev. Lett.\/} {\bf 116}(15) 150502
  \urlprefix\url{http://link.aps.org/doi/10.1103/PhysRevLett.116.150502}

\bibitem{Yao15}
Yao Y, Xiao X, Ge L and Sun C~P 2015 {\em Phys. Rev. A\/} {\bf 92}(2) 022112
  \urlprefix\url{http://link.aps.org/doi/10.1103/PhysRevA.92.022112}

\bibitem{Streltsov15}
Streltsov A, Singh U, Dhar H~S, Bera M~N and Adesso G 2015 {\em Phys. Rev.
  Lett.\/} {\bf 115}(2) 020403
  \urlprefix\url{https://link.aps.org/doi/10.1103/PhysRevLett.115.020403}

\bibitem{Streltsov16}
Streltsov A, Chitambar E, Rana S, Bera M~N, Winter A and Lewenstein M 2016 {\em
  Phys. Rev. Lett.\/} {\bf 116}(24) 240405
  \urlprefix\url{https://link.aps.org/doi/10.1103/PhysRevLett.116.240405}

\bibitem{ma16}
Ma J, Yadin B, Girolami D, Vedral V and Gu M 2016 {\em Phys. Rev. Lett.\/} {\bf
  116}(16) 160407
  \urlprefix\url{https://link.aps.org/doi/10.1103/PhysRevLett.116.160407}

\bibitem{Chitambar16}
Chitambar E and Hsieh M~H 2016 {\em Phys. Rev. Lett.\/} {\bf 117}(2) 020402
  \urlprefix\url{https://link.aps.org/doi/10.1103/PhysRevLett.117.020402}

\bibitem{Hu17relative}
Hu M~L and Fan H 2017 {\em Phys. Rev. A\/} {\bf 95}(5) 052106
  \urlprefix\url{https://link.aps.org/doi/10.1103/PhysRevA.95.052106}

\bibitem{PhysRevX.7.011024}
Streltsov A, Rana S, Bera M~N and Lewenstein M 2017 {\em Phys. Rev. X\/} {\bf
  7}(1) 011024
  \urlprefix\url{https://link.aps.org/doi/10.1103/PhysRevX.7.011024}

\bibitem{yuan2017unified}
Yuan X, Zhou H, Gu M and Ma X 2017 {\em arXiv preprint arXiv:1706.04853\/}

\bibitem{yu2016measure}
Yu X~D, Zhang D~J, Liu C~L and Tong D~M 2016 {\em Phys. Rev. A\/} {\bf 93}(6)
  060303 \urlprefix\url{https://link.aps.org/doi/10.1103/PhysRevA.93.060303}

\bibitem{Yu16Alternative}
Yu X~D, Zhang D~J, Xu G~F and Tong D~M 2016 {\em Phys. Rev. A\/} {\bf 94}(6)
  060302 \urlprefix\url{https://link.aps.org/doi/10.1103/PhysRevA.94.060302}

\bibitem{PhysRevX.6.041028}
Yadin B, Ma J, Girolami D, Gu M and Vedral V 2016 {\em Phys. Rev. X\/} {\bf
  6}(4) 041028
  \urlprefix\url{https://link.aps.org/doi/10.1103/PhysRevX.6.041028}

\bibitem{Shi17}
Shi H~L, Liu S~Y, Wang X~H, Yang W~L, Yang Z~Y and Fan H 2017 {\em Phys. Rev.
  A\/} {\bf 95}(3) 032307
  \urlprefix\url{https://link.aps.org/doi/10.1103/PhysRevA.95.032307}

\bibitem{hu2017maximum}
{Hu} M~L, {Shen} S~Q and {Fan} H 2017 {\em ArXiv e-prints\/} (\textit{Preprint}
  \eprint{1707.09617})

\bibitem{shi2017coherence}
{Shi} H~L, {Wang} X~H, {Liu} S~Y, {Yang} W~L, {Yang} Z~Y and {Fan} H 2017 {\em
  ArXiv e-prints\/} (\textit{Preprint} \eprint{1705.00785})

\bibitem{Streltsov17structure}
Streltsov A, Rana S, Boes P and Eisert J 2017 {\em Phys. Rev. Lett.\/} {\bf
  119}(14) 140402
  \urlprefix\url{https://link.aps.org/doi/10.1103/PhysRevLett.119.140402}

\bibitem{Zhou17}
Zhou Y, Zhao Q, Yuan X and Ma X 2017 {\em New Journal of Physics\/}
  \urlprefix\url{http://iopscience.iop.org/10.1088/1367-2630/aa91fa}

\bibitem{Yang2017}
ren Yang S and shui Yu C 2017 {\em Annals of Physics\/}  -- ISSN 0003-4916
  \urlprefix\url{https://www.sciencedirect.com/science/article/pii/S0003491617303500}

\bibitem{streltsov2016quantum}
Streltsov A, Adesso G and Plenio M~B 2017 {\em Rev. Mod. Phys.\/} {\bf 89}(4)
  041003 \urlprefix\url{https://link.aps.org/doi/10.1103/RevModPhys.89.041003}

\bibitem{Hu2017arXivreview}
{Hu} M~L, {Hu} X, {Wang} J~C, {Peng} Y, {Zhang} Y~R and {Fan} H 2017 {\em ArXiv
  e-prints\/} (\textit{Preprint} \eprint{1703.01852})

\bibitem{yuan2016interplay}
{Yuan} X, {Zhao} Q, {Girolami} D and {Ma} X 2016 {\em ArXiv e-prints\/}
  (\textit{Preprint} \eprint{1605.07818})

\bibitem{hayashi2017secure}
{Hayashi} M and {Zhu} H 2017 {\em ArXiv e-prints\/} (\textit{Preprint}
  \eprint{1706.04009})

\bibitem{Yuan17uncertainty}
Yuan X, Bai G, Peng T and Ma X 2017 {\em Phys. Rev. A\/} {\bf 96}(3) 032313
  \urlprefix\url{https://link.aps.org/doi/10.1103/PhysRevA.96.032313}

\bibitem{Luo17uncertainty}
Luo S and Sun Y 2017 {\em Phys. Rev. A\/} {\bf 96}(2) 022130
  \urlprefix\url{https://link.aps.org/doi/10.1103/PhysRevA.96.022130}

\bibitem{ma2017source}
Ma J, Yuan X, Hakande A and Ma X 2017 {\em arXiv preprint arXiv:1704.06915\/}

\bibitem{Rastegin16}
Rastegin A~E 2016 {\em Phys. Rev. A\/} {\bf 93}(3) 032136
  \urlprefix\url{http://link.aps.org/doi/10.1103/PhysRevA.93.032136}

\bibitem{zhao2017one}
{Zhao} Q, {Liu} Y, {Yuan} X, {Chitambar} E and {Ma} X 2017 {\em ArXiv
  e-prints\/} (\textit{Preprint} \eprint{1707.02522})

\bibitem{chitambar2016comparison}
Chitambar E and Gour G 2016 {\em Phys. Rev. A\/} {\bf 94}(5) 052336
  \urlprefix\url{https://link.aps.org/doi/10.1103/PhysRevA.94.052336}

\bibitem{Winter16}
Winter A and Yang D 2016 {\em Phys. Rev. Lett.\/} {\bf 116}(12) 120404
  \urlprefix\url{http://link.aps.org/doi/10.1103/PhysRevLett.116.120404}

\bibitem{streltsov2015measuring}
Streltsov A, Singh U, Dhar H~S, Bera M~N and Adesso G 2015 {\em Phys. Rev.
  Lett.\/} {\bf 115}(2) 020403
  \urlprefix\url{https://link.aps.org/doi/10.1103/PhysRevLett.115.020403}

\bibitem{2013MartinRenyi}
MÃ¼ller-Lennert M, Dupuis F, Szehr O, Fehr S and Tomamichel M 2013 {\em
  Journal of Mathematical Physics\/} {\bf 54} 122203 (\textit{Preprint}
  \eprint{https://doi.org/10.1063/1.4838856})
  \urlprefix\url{https://doi.org/10.1063/1.4838856}

\bibitem{coles2012unification}
Coles P~J 2012 {\em Phys. Rev. A\/} {\bf 85}(4) 042103
  \urlprefix\url{https://link.aps.org/doi/10.1103/PhysRevA.85.042103}

\bibitem{konig2009operational}
Konig R, Renner R and Schaffner C 2009 {\em IEEE Transactions on Information
  Theory\/} {\bf 55} 4337--4347 ISSN 0018-9448

\bibitem{zhang2017estimation}
Zhang H~J, Chen B, Li M, Fei S~M and Long G~L 2017 {\em Communications in
  Theoretical Physics\/} {\bf 67} 166
  \urlprefix\url{http://stacks.iop.org/0253-6102/67/i=2/a=166}

\bibitem{tomamichel2012framework}
Tomamichel M 2012 {\em ArXiv e-prints\/} (\textit{Preprint} \eprint{1203.2142})

\bibitem{wang2012one}
Wang L and Renner R 2012 {\em Phys. Rev. Lett.\/} {\bf 108}(20) 200501
  \urlprefix\url{https://link.aps.org/doi/10.1103/PhysRevLett.108.200501}

\bibitem{2017arXiv171110512R}
{Regula} B, {Fang} K, {Wang} X and {Adesso} G 2017 {\em ArXiv e-prints\/}
  (\textit{Preprint} \eprint{1711.10512})

\bibitem{Chitambar16prl}
Chitambar E and Gour G 2016 {\em Phys. Rev. Lett.\/} {\bf 117}(3) 030401
  \urlprefix\url{https://link.aps.org/doi/10.1103/PhysRevLett.117.030401}

\bibitem{Marvian16}
Marvian I and Spekkens R~W 2016 {\em Phys. Rev. A\/} {\bf 94}(5) 052324
  \urlprefix\url{http://link.aps.org/doi/10.1103/PhysRevA.94.052324}

\bibitem{Buscemi2011prlentanglement}
Buscemi F and Datta N 2011 {\em Phys. Rev. Lett.\/} {\bf 106}(13) 130503
  \urlprefix\url{https://link.aps.org/doi/10.1103/PhysRevLett.106.130503}

\bibitem{Buscemi2010distill}
Buscemi F and Datta N 2010 {\em Journal of Mathematical Physics\/} {\bf 51}
  102201 (\textit{Preprint} \eprint{https://doi.org/10.1063/1.3483717})
  \urlprefix\url{https://doi.org/10.1063/1.3483717}

\bibitem{Devetak2005distill}
Devetak I and Winter A 2005 {\em Proceedings of the Royal Society of London A:
  Mathematical, Physical and Engineering Sciences\/} {\bf 461} 207--235 ISSN
  1364-5021 (\textit{Preprint}
  \eprint{http://rspa.royalsocietypublishing.org/content/461/2053/207.full.pdf})
  \urlprefix\url{http://rspa.royalsocietypublishing.org/content/461/2053/207}

\bibitem{renner2008security}
RENNER R 2008 {\em International Journal of Quantum Information\/} {\bf 06}
  1--127 (\textit{Preprint}
  \eprint{http://www.worldscientific.com/doi/pdf/10.1142/S0219749908003256})
  \urlprefix\url{http://www.worldscientific.com/doi/abs/10.1142/S0219749908003256}

\end{thebibliography}

\end{document}